\documentclass{article}

\usepackage{amsfonts,amssymb,bbm}
\usepackage[english]{babel}
\usepackage{color}
\usepackage{array}
\usepackage[T1]{fontenc}

\def\QSCDff{\mathbbm{QSCD}_{f\!f}}

\def\sgn{sgn}
\newcommand{\bra}[1]{{\left\langle{#1}\right\vert}}
\newcommand{\ket}[1]{{\left\vert{#1}\right\rangle}}
\newcommand{\bracket}[2]{\left\langle #1| #2 \right\rangle}
\def\nats{\mathbbm N}
\newcommand\K{\mathbbm{K}}
\renewcommand\S{\mathbbm{S}}
\def\E{\mathbbm{E}}
\def\O{\mathbbm{O}}
\def\ds{\displaystyle}

\newcommand\cpi{C_\pi}

\newcommand\cone{C_1}
\newcommand{\I}{\mathbbm 1}

\def\qed{\hfill $\square$ }

\newtheorem{theorem}{Theorem}[section]
\newtheorem{proposition}[theorem]{Proposition}
\newtheorem{lemma}[theorem]{Lemma}

\newtheorem{algorithm}[theorem]{Algorithm}
\newtheorem{protocol}[theorem]{Protocol}

\newtheorem{example}[theorem]{Example}

\newtheorem{definition}[theorem]{Definition}

\newtheorem{problem}[theorem]{Problem}

\newenvironment{proof}{\begin{trivlist}\item{\bf Proof:}}{\end{trivlist}}

\def\makenewenum#1#2{%
\newcounter{cnt#1}
\newenvironment{#1}%
{\begin{list}{\makebox[0pt][r]{#2}}%
{\setlength{\itemsep}{0pt}%
 \setlength{\parsep}{.2em}%
 \setlength{\leftmargin}{2.5em}%
 \setlength{\labelwidth}{.2em}%
 \usecounter{cnt#1}}}%
{\end{list}}}

\makenewenum{senum}{\bf{Step~\arabic{cntsenum}.}}

\usepackage[affil-it]{authblk}
\begin{document}
\title
{Oblivious transfer based on quantum state computational distinguishability}

\author{	A Souto$^{1,2}$, 
		P Mateus$^{1,2}$, 
		P Ad\~ao$^{1,3}$ and 
		N Paunkovi\'c$^{1,2}$
	}
\affil{$^1$ SQIG---Instituto de Telecomunica\c{c}\~oes -- Avenida Rovisco Pais 1, 1049-001 Lisbon, Portugal }
\affil{$^2$ Departamento de Matem\'atica, Instituto Superior T\'ecnico, Universidade de Lisboa -- Avenida Rovisco Pais 1, 1049-001 Lisbon, Portugal }
\affil{$^3$ Departamento de Engenharia Inform\'atica, Instituto Superior T\'ecnico, Universidade de Lisboa -- Avenida Rovisco Pais 1, 1049-001 Lisbon, Portugal }
\affil{		\bf{a.souto@math.ist.utl.pt}, 
		\bf{pmat@math.ist.utl.pt},  
		\bf{pedro.adao@ist.utl.pt}, 
		\bf{npaunkovic@math.ist.utl.pt}}
		

\maketitle

\begin{abstract}
Oblivious transfer protocol is a basic building block in cryptography and is used to transfer information from a sender to a receiver in such a way that, at the end of the protocol, the sender does not know if the receiver got the message or not.

Since Shor's quantum algorithm appeared,
the security of most of classical cryptographic schemes has been compromised, as they rely on the fact that factoring is unfeasible. 
To overcome this, quantum mechanics has been used intensively in the past decades, and alternatives resistant to quantum attacks have been developed in order to fulfill the (potential) lack of security of a significant number of classical schemes.

In this paper, we present a quantum computationally secure protocol for oblivious transfer between two parties, under the assumption of quantum hardness of state distinguishability. The protocol is feasible, in the sense that it is implementable in polynomial time.
\end{abstract} 


\section{Introduction}
An oblivious transfer protocol involves two parties, a sender (Alice) and a receiver (Bob). It consists of two phases: the transferring phase and the opening phase.
The goal of the sender is to send a message during the transferring phase, that will not be known to the receiver until the opening phase, during which the sender reveals the message with probability $1/2$. 
The goal of the receiver is that, upon opening the message, the sender is oblivious to the fact that the message was successfully transferred or not. 
Although not explicitly stated in the original argument, it is usually assumed that the receiver knows, at the end of the protocol, if he got the correct message (see for example \cite{cre:87}).

The first oblivious transfer scheme was proposed by Rabin~\cite{rab:81} and is based on the same assumptions as the RSA cryptographic system.  
In Rabin's scheme, the sender sends a message to the receiver that is able to decrypt it properly with probability 1/2. 
At the end of the protocol, the sender remains oblivious to whether or not the receiver got the correct message. 
Cr\'epeau later showed that Rabin's oblivious transfer is equivalent to $1$-out-of-$2$ oblivious transfer \cite{cre:87}.  
In $1$-out-of-$2$ oblivious transfer the sender has two messages to send such that:  the receiver gets only one of the two with equal probability; the sender is oblivious to which message was received. 
Although not considered by then in the cryptographic domain, Wiesner had already proposed (in the 70's but only published in the 80's) in his pioneer work on quantum cryptography a scheme based on multiplexing which is equivalent to $1$-out-of-$2$ oblivious transfer protocol \cite{wie:83}.
Even, Goldreich, and Lempel~\cite{eve:gol:lem:85} generalized the original idea of $1$-out-of-$2$ protocol proposing the notion of $1$-out-of-$n$ oblivious transfer protocol.
In this case, the receiver gets $1$-out-of-$n$ possible different messages.

Oblivious transfer is a particularly important primitive as from an oblivious transfer protocol one can apply the technique presented in \cite{cre:87} to construct a $1$-out-of-$2$ oblivious transfer protocol which in turn, using the technique presented in~\cite{ben:bra:cre:sku:91}, can be transformed into a bit-commitment scheme from $O(n)$ oblivious transfer instances. 
Furthermore, practical protocols of oblivious transfer are useful in designing secure multiparty computation schemes, as suggested in recent papers \cite{lin:pin:12,lin:zar:13}, and are building blocks for more complex cryptographic protocols~\cite{bra:cre:rob:86,kil:88,har:lin:93,may:95,cra:dam:mau:00}.

In the last decades quantum computation has played a crucial role in the development of cryptographic analysis. 
The breakthrough of quantum computation in the realm of cryptography is due to Shor's factoring algorithm~\cite{sho:97}.
This algorithm compromises the security of most common public key cryptographic schemes as factoring becomes feasible with a quantum computer. 
Since today's technology is evolving to be able to deal with a larger number of qubits, a possibility of having affordable and reliable quantum computers in the future arises, compromising the secrecy of communications and transactions.
To overcome this possibility, researchers have focused on the development of protocols that are resilient to quantum adversaries. This effort started with key distribution protocol proposed in \cite{ben:bra:84} that was shown to be secure by Shor and Preskill \cite{sho:pre:00}. This result boosted the development of other types of quantum cryptographic protocols (see for instance~\cite{chi:08,mat:wat:10}).
One such example is the quantum version oblivious transfer proposed by Bennet, Brassard, Cr\'epeau and Skubiszewska~\cite{ben:bra:cre:sku:91}.
However, as perfectly secure (quantum) bit commitment and (quantum) oblivious transfer are known to be impossible~\cite{may:97,lo:chau:96}, 
all the schemes proposed for these primitives will necessarily assume a tradeoff between cheating strategies or will require some computational security assumptions.

In this paper we present a polynomial time quantum oblivious transfer protocol based on a presumed polynomial hard problem even for a quantum computer,\footnote{A problem is polynomially hard if there is no polynomial time algorithm for solving it.} the {\em Quantum State Computational Distinguishability with Fully Flipped Permutations}, denoted by $\QSCDff$, presented in \cite{kaw:kos:nis:yam:05}.  
It is the standard problem of quantum state distinguishability, applied to two particular quantum states: $\otimes_{i=1}^{k(n)}\frac{1}{\sqrt{2}}(\ket{\sigma_i}+\ket{\sigma_i\circ\pi})$ and $\otimes_{i=1}^{k(n)}\frac{1}{\sqrt{2}}(\ket{\sigma_i}-\ket{\sigma_i\circ\pi})$, where $\sigma_i\in\S_n,\pi\in\K_n$ are permutations ($\K_n$ is the set of all permutation in $\S_n$  of order 2, such that $\pi(i)\neq i$ for all $i$), and $k$ is some polynomial, such that for each state the array of  $\sigma_i$s is chosen at random. 
In the same paper, the authors discuss how one can explore the indistinguishability of these quantum states in the scope of cryptography and propose a cryptographic public key scheme that, under the hardness assumption of $\QSCDff$, is secure against polynomial quantum adversaries using $\pi$ as a trapdoor.

The rest of this paper is organized as follows: 
in the next section we present the basic notions, problems, notation and results used in the rest of the paper. 
We also show the equivalence between a previously known algorithm, used in \cite{kaw:kos:nis:yam:05}, and a particular orthogonal measurement that is used to prove the security of our protocol, under the assumption that $\QSCDff$ is polynomially hard for quantum computers. 	
In Section~\ref{sec:protocol} we present our protocol for oblivious transfer and prove its correctness and security.
In Section~\ref{sec:conclusion} we present the conclusions.

\section{Preliminaries}
We use a binary alphabet $\Sigma = \{0,1\}$ and strings of length $\ell$, which are elements of $\Sigma^\ell$.
The group $\S S_n$ is the set of all permutations over the set $\{1,\dots,n\}$, whose elements we denote by Greek letters $\sigma,\tau,\delta,\dots$, together with the composition operation $\circ$.\footnote[1]{Formally, each permutation $\sigma\in \S_n$ is a bijective function over $\{1,\dots,n\}$.}
In the rest of the paper we assume $n$ to be of the form $2(2m+1)$, for some $m\in\mathbbm{N}$.
\begin{example}\label{ex:one}
	Consider $n=6$ and $\sigma\in \S_6$ defined as 
	$\sigma(1)=2,\sigma(2)=3,\sigma(3)=1,\sigma(4)=5,\sigma(5)=4$ and $\sigma(6)=6$. 
	We represent this permutation as $\sigma=(1\ 2\ 3)(4\ 5)$, where $(1\ 2 \ 3)$ and $(4\ 5)$ represent the orbits of elements of $\{1,\dots 6\}$. The orbit of $i\in\{1,\dots,n\}$ with respect to $\sigma$ is $(i\ \sigma(i)\ \dots\ \sigma^j(i))$, where the superscript is the number of times $\sigma$ is applied to element $i$, and $j$ is the smallest integer such that $\sigma(\sigma^j(i))=i$.
	
	Notice that this representation is not unique, and we can represent this same $\sigma$ as 
	$(4\ 5)(1\ 2\ 3)$ or $(5\ 4)(2\ 3\ 1)$.
\end{example}

Given a permutation $\sigma\in\S_n$, other than the identity, one can decompose it into a sequence of transpositions, \textit{i.e.}, elementary permutations that only exchange two elements.
It is easy to see that such decomposition is not unique, but the number of transpositions, denoted by $\#(\sigma)$, has always the same {\em parity}  and hence one can define the sign of a permutation $\sigma$ as 
$\sgn(\sigma)=(-1)^{\#(\sigma)}$. 
\begin{example}
	Consider $\sigma$ as defined in Example~\ref{ex:one}. 
	Three possible decompositions of $\sigma$ in terms of transpositions (derived from the three given representations) are
	$(1\ 3)(1\ 2)(4\ 5)$ and 
	$(4\ 5)(1\ 3)(1\ 2)$ and 
	$(5\ 4)(2\ 1)(2\ 3)$, 
	and all of them have parity~1.
	
	One can, in fact, show that if a permutation $\sigma$ generates $L$ orbits of elements from $\{1,\dots,n\}$  of lengths 
	$\ell_1,\dots,\ell_L$ 
	then $\#(\sigma)=\sum_{i=1}^L (\ell_i-1)$.
	In our case, the representation of $\sigma$ in Example~\ref{ex:one} has orbits of length 3 and 2, 
	hence $\#(\sigma)=(3-1)+(2-1)=3$. 
	The same $\sigma$ was represented in this example with 3 orbits of length 2, hence 
	$\#(\sigma)=3\times(2-1)=3$. 
\end{example}

Since $|\S_n|=n!$ one needs $\log(n!) = \sum_{i=1}^n\log i \leq n\log n\in O(n\log n)$ bits to represent each $\sigma\in\S_n$. Note that $\S_n$ consists of two sets of equal size: 
$\E_n$ containing the even permutations (i.e., permutations with sign $1$), and its complement $\O_n$ consisting of all odd permutations. 
Hence $\S_n= \E_n\cup \O_n$. 

As in~\cite{kaw:kos:nis:yam:05}, we consider the following subset of $\S_ n$:
$$
\K_ n =
		\Big\{
				\pi\in \S_n : 
				\pi\circ\pi =id_n\textnormal{ and } 
				\pi(i)\neq i, \textnormal{ for all } i\in\{1,\cdots, n\}
		\Big\}.
$$ 

\begin{example}
	One can see that the permutation $\sigma$ of Example~\ref{ex:one} is not in $\K_6$ as 
	$\sigma\circ\sigma \neq id_n$, \textit{e.g.}, $\sigma\circ\sigma (1)=\sigma(2)=3\neq 1$.
	We also do not have $\sigma(i)\neq i$ for $i=6$.
	
	As an example of a permutation in $\K_6$ we have 
	$\pi=(1\ 2)(3\ 6)(4\ 6)$.
\end{example}

One can immediately see that for an even $n$, any permutation $\pi\in\K_n$ can always be decomposed into $n/2$~transpositions since 
each element $i\in\{1,\dots,n\}$ has to appear in (at least) one transposition (otherwise $\pi(i)= i$),
and appears only once, as $\pi^2=id_n$.
Given that we assumed $n=2(2m+1)$, we have that $\pi\in\K_ n$ if and only if it can be written as an odd number of transpositions and hence $\pi\in \O_n$.

A counting argument revels that the size of $\K_ n$ is:  
$$
|\K_n| 
		= \ds\left(\begin{array}{c}n\\ n/2\end{array}\right) \left(\ds \frac n 2\right)!
		=\ds{\frac{n!}{\ds \left(\frac n 2\right)!}}.
$$ 

Let $n\in \mathbbm N$ and $\pi\in \K_n$, and consider the Hilbert space $$\mathcal H_n = span\{\ket \sigma: \sigma\in \S_n\},$$ such that for all $\sigma$ and $\sigma'$, $\bracket \sigma{\sigma'}=\delta_{\sigma,\sigma'}$. 
Let $\ket{\psi^\pm_{\pi} (\sigma)} = \frac{1}{\sqrt{2}} (\ket\sigma \pm \ket{\sigma\circ\pi})$.
Notice that for every $\pi$ the set 
		$\{\ket{\psi^\pm_{\pi} (\sigma)}:\sigma\in\E_n\}$
is an orthonormal basis.

Indeed,  we have that 
$$\bracket{\psi^\pm_{\pi} (\sigma)}{\psi^\pm_{\pi} (\sigma')}= 
		\frac 1 2 \left(\bracket{\sigma}{\sigma'} \pm \bracket{\sigma}{\sigma'\circ\pi} \pm \bracket{\sigma\circ\pi}{\sigma'} + \bracket{\sigma\circ\pi}{\sigma'\circ\pi}\right)$$ 
is either $1$, when $\sigma=\sigma'$, or $0$ otherwise.
This is because $\sigma,\sigma'\in\E_n$ and $\pi\in\O_n$, and consequently $\bracket{\sigma}{\sigma'\circ\pi} = \bracket{\sigma\circ\pi}{\sigma'} = 0$. 

On the other hand 
$$\bracket{\psi^\pm_{\pi} (\sigma)}{\psi^\mp_{\pi} (\sigma')}= 
		\frac 1 2 \left(\bracket{\sigma}{\sigma'} 
			\mp \bracket{\sigma}{\sigma'\circ\pi} 
			\pm \bracket{\sigma\circ\pi}{\sigma'} 
			- \bracket{\sigma\circ\pi}{\sigma'\circ\pi}\right)$$ 
is always $0$:
when $\sigma=\sigma'$, the first and last terms cancel each other, while the other two terms are always zero. Since
$\ket{\psi^\pm_{\pi} (\sigma)}= \pm\ket{\psi^\pm_{\pi} (\sigma\circ\pi)}$, the set $\{\ket{\psi^\pm_{\pi} (\sigma)}:\sigma\in\O_n\}$ is also an orthonormal basis.

Consider the following quantum states defined in~\cite{kaw:kos:nis:yam:05}:
$$
\begin{array}{rcl}
\rho^+_\pi
	&=& \ds \frac{1}{2n!} \sum_{\sigma \in \S_n}(\ket\sigma + \ket{\sigma \circ\pi})(\bra\sigma + \bra{\sigma \circ\pi}) \\[10pt]
	&=& \ds  \frac{1}{n!}\sum_{\sigma \in \S_n} \ket{\psi^+_{\pi} (\sigma)}\bra{\psi^+_{\pi} (\sigma)}\\[10pt]
	&=& \ds\frac{2}{n!}\sum_{\sigma \in \E_n} \ket{\psi^+_{\pi} (\sigma)}\bra{\psi^+_{\pi} (\sigma)}
\end{array}
$$
and
$$
\begin{array}{rcl}
\rho^-_\pi
	&=& \ds \frac{1}{2n!} \sum_{\sigma \in \S_n}(\ket\sigma - \ket{\sigma \circ\pi})(\bra\sigma - \bra{\sigma \circ\pi})\\[10pt]
	&=& \ds\frac{1}{n!}\sum_{\sigma \in \S_n} \ket{\psi^-_{\pi} (\sigma)}\bra{\psi^-_{\pi} (\sigma)}\\[10pt]
	&=& \ds \frac{2}{n!}\sum_{\sigma \in \E_n} \ket{\psi^-_{\pi} (\sigma)}\bra{\psi^-_{\pi} (\sigma)}.
\end{array}
$$ 

We are interested in these particular states as they are orthogonal to each other and hence fully distinguishable, {\it provided one knows} which $\pi$ was used to prepare them. 
Without the knowledge of $\pi$ the problem of distinguishing these states is believed to be polynomially hard even for quantum computers, as stated in~\cite{kaw:kos:nis:yam:05}. First, we state the problem:

\begin{problem}[$\QSCDff$]
The {\em Quantum State Computational Distinction with Fully Flipped Per\-mutations Problem}, denoted by $\QSCDff$, is defined as:
\begin{description}
	\item[\hspace{5mm} Instances:] 
		Two quantum states 
			$(\rho^+_\pi)^{\otimes k(n)}$ and $(\rho^-_\pi)^{\otimes k(n)}$ 
		where $n=2(2m+1)$ for some $m\in \mathbbm N$, and $k$ is some fixed polynomial, 
		\textit{i.e.}, each state consists of $k(n)$ copies of $\rho^+_\pi$ and $\rho^-_\pi$, respectively. 
	
	\item[\hspace{5mm} Question:] 
		Are 
			$(\rho^+_\pi)^{\otimes k(n)}$ and $(\rho^-_\pi)^{\otimes k(n)}$ 
		computationally indistinguishable, 
		\textit{i.e.}, is the probability of a quantum polynomial time algorithm to be able to distinguish between the states $(\rho^+_\pi)^{\otimes k(n)}$ and $(\rho^-_\pi)^{\otimes k(n)}$ a negligible function?
\end{description}
\end{problem}

The problem $\QSCDff$ is closely related to the hidden subgroup problem over symmetric groups for which no one knows an efficient quantum algorithm to solve it. 
In~\cite{kaw:kos:nis:yam:05} the authors reduced the $\QSCDff$ problem to a variant of the unique graph automorphism problem, that is also presumably hard for quantum computers with polynomial time resources, proving the following hardness result:

\begin{theorem}[\cite{kaw:kos:nis:yam:05}]\label{theo:hardness}
	If there exists a polynomial-time quantum algorithm that solves $\QSCDff$ with non-negligible advantage, then there exists a polynomial-time quantum algorithm that solves the graph automorphism problem in the worst case for infinitely-many input lengths.
\end{theorem}

The interesting property that makes this problem suitable for cryptography is that it has a trapdoor that allows one to efficiently distinguish the states: one can distinguish, with certainty, 
	$(\rho^+_\pi)^{\otimes k(n)}$ from $(\rho^-_\pi)^{\otimes k(n)}$, 
provided that an extra piece of information is given, in this case $\pi$. 
Furthermore, when using a permutation $\pi'$ different from the trapdoor $\pi$, the probability of distinguishing these states is the same as plain guessing. In order to present the protocol and analyze its complexity and security, we need the following obvious proposition:

\newcommand\csignal{C_{\textit{sgn}}}
\newcommand\cswap{C_{\textit{swap}}}

\begin{proposition}\label{prop:operations}
	The following linear operators are unitary:
	\begin{itemize}
	\item $\cpi((\ket 0+\ket 1) \ket\sigma)=\ket 0 \ket{\sigma}+ \ket 1 \ket{\sigma\circ\pi}$;
	
	\item $\cone(\ket 0 \ket\sigma +\ket 1 \ket\pi)=\ket 0( \ket{\sigma}+  \ket{\pi})$;
	
	\item $C^r_\circ(\ket\phi\ket\psi)=\ket\phi\ket{\phi\circ\psi}$ and $C^l_\circ(\ket\phi\ket\psi)=\ket\phi\ket{\psi\circ\phi}$;
	
	\item $\cswap(\ket\phi \ket\psi)=\ket\psi \ket\phi$;
	
	\item $\csignal(\ket\alpha)=(-1)^{\sgn(\alpha)}\ket\alpha$.
	\end{itemize}
\end{proposition}

The quantum algorithm presented in~\cite{kaw:kos:nis:yam:05}, Algorithm \ref{alg:keysgeneration}, justifies that given the private key $\pi$ one can produce the states $\ds (\rho^+_\pi)^{\otimes k(n)}$ 
in {\em polynomial time} 
using only Hadamard $H$ and the operations defined in Proposition \ref{prop:operations}.

\begin{algorithm}
	[to generate $\ds (\rho^+_\pi)^{\otimes k(n)}$~\cite{kaw:kos:nis:yam:05}]
	\label{alg:keysgeneration}
	For the sake of simplicity of the presentation we consider a fixed $n$ and the case where $k(n)=1$. The reader can easily generalize the argument for $k(n)$ systems.
	\begin{quote}
		\begin{description}
			\item[Input:] $\pi\in \K_n$.
			\item[Output:] $\ds \rho^+_\pi= \frac{2}{n!} \sum_{\sigma\in \E_n}(\ket \sigma+ \ket{\sigma \circ\pi})(\bra \sigma+ \bra{\sigma\circ \pi})$. 
				\begin{senum}
					\item Select a random $\sigma\in\S_n$ and prepare the initial state 
							$\ket 0\otimes\ket{id_n}\otimes\ket{\sigma} \in \mathbbm C^2\otimes \mathbbm C^{n\log n}\otimes \mathbbm C^{n\log n}$.
					\item Apply $H\otimes\I\otimes\I$.
					\item Apply $\cpi\otimes\I$.
					\item Apply $\cone\otimes\I$.
					\item Apply $\I\otimes\cswap$.
					\item Apply $\I\otimes C^r_\circ$.
				\end{senum}
		\end{description}
	\end{quote}
	The third register is now in the state
	$\ket{\psi^+_{\pi} (\sigma)} = \frac{1}{\sqrt{2}} (\ket\sigma + \ket{\sigma\circ\pi})$. Since $\sigma$ is chosen at random, the ensemble of such systems is represented by the mixed state $\rho^+_\pi$.
\end{algorithm}

Note that it is possible to prepare the state $\rho^+_\pi$ as a partial trace of an entangled state.  
In order to do that, one should start from 
$\ket 0 \ket{id_n} \frac{1}{\sqrt{n!}} \sum_{\sigma\in\S_n}\ket\sigma$. 
Upon applying Steps 1 to 6 of the previous algorithm the overall state is 
$\ket 0 \frac{1}{\sqrt{n!}} \sum_{\sigma\in\S_n}\ket\sigma \frac{1}{\sqrt{2}}(\ket\sigma + \ket{\sigma\circ\pi})$ and the third register is again in the state $\rho^+_\pi$. 

In the next lemma we show that one can easily transform $\rho_\pi^+$ into $\rho_\pi^-$, and vice versa, without knowing $\pi$. This property follows immediately from the fact that $\pi$ is an odd permutation. 

\begin{lemma}[\cite{kaw:kos:nis:yam:05}]
\label{alg:conversion}
There exists a {\em polynomial-time} quantum algorithm that, with probability $1$, transforms $\ds\rho_\pi^+$ into $\rho_\pi^-$ and keeps $ \I = \frac{1}{n!} \sum_{\sigma\in\S_n} \ket\sigma\bra\sigma$ invariant, for any $n=2(2m+1)$, with $m \in\nats$, and any permutation $\pi\in\K_n$.
\end{lemma}

\begin{proof}
Let $\pi\in\K_n$ be a permutation and consider a pure quantum state $\ket{\psi^+_{\pi} (\sigma)}$. The desired algorithm  implements the transformation:
$$\begin{array}{rcl}
\ket{\psi^+_{\pi} (\sigma)} 
	&=& \ds\frac{\ket\sigma +\ket{\sigma \circ \pi}}{\sqrt{2}} \\ 
	& \ds\stackrel{\csignal}{\longrightarrow} &\ds \frac{(-1)^{\sgn(\sigma)}\ket\sigma+(-1)^{\sgn(\sigma)+1}\ket{\sigma \circ \pi} }{\sqrt{2}}\\
	&=& (-1)^{\sgn(\sigma)}|\psi^-_{\pi} (\sigma)\rangle.
\end{array}$$
Notice that determining the sign of a permutation is a computation that can be done in polynomial time. Furthermore, it is easy to see that this algorithm leaves $\I$ invariant.\qed
\end{proof}

The following algorithm explores the trapdoor property of $\QSCDff$: given~$\pi$, 
one can distinguish in {\em polynomial time} and with probability 1 the quantum states 
$(\rho^+_\pi)^{\otimes k(n)}$ 
and 
$(\rho^-_\pi)^{\otimes k(n)}$.

\begin{algorithm}
	[to distinguish $(\rho^+_\pi)^{\otimes k(n)}$ from $(\rho^-_\pi)^{\otimes k(n)}$~\cite{kaw:kos:nis:yam:05}]
	\label{alg:distinction}
	For the sake of simplicity of the presentation we consider a fixed $n$ and the case where $k(n)=1$. 
	The reader can easily generalize the argument for $k(n)$ systems.
	\begin{quote}
		\begin{description}
			\item[Input:]  	$\pi\in\K_n$ and a quantum state 
						$\chi_\pi$ that is either $\rho^+_\pi$ or $\rho^-_\pi$. 
			\item[Output:] $0$ if $\chi_\pi =  \rho^+_\pi$, and $1$ if $\chi_\pi =  \rho^-_\pi$.
				\begin{senum}
					\item Prepare the system in the state $\ket 0 \bra 0\otimes \chi_\pi$, where $\ket 0 \in\mathbbm C^2$.
					\item Apply $H\otimes \I$.
					\item Apply $C_\pi$.\label{prot:distinction-step3}
					\item Apply again $H\otimes \I$. 
					\item Measure $M_{+} =(0\cdot \ket 0 \bra 0 + 1\cdot\ket 1\bra 1)\otimes \I$ and output the result. \label{step:5:alg}
				\end{senum}
		\end{description}
	\end{quote}
\end{algorithm}
\begin{proof}
After performing the second step the obtained state is 
$\frac{1}{2}\ket + \bra + \otimes \chi_\pi$, 
with 
$\ket\pm = \frac{1}{\sqrt 2}(\ket{0} \pm \ket 1)$. 
Since 
$\cpi(\ket +(\ket \sigma  \pm \ket{\sigma\circ\pi}))
		= \ket \pm(\ket \sigma  \pm \ket{\sigma\circ\pi})$, 
the overall state after the third step is 
$\ket\pm\bra\pm  \otimes \rho^\pm_\pi$.
So, by applying Hadamard on the first register and measuring it in the computational basis we get the desired outcome. \qed
\end{proof}

We have shown that the above algorithm is able to distinguish with certainty between 
	$\chi_\pi=(\rho^+_\pi)^{\otimes k(n)}$ and $\chi_\pi=(\rho^-_\pi)^{\otimes k(n)}$, using only polynomial quantum resources, 
provided that we use the correct $\pi$.
We will show later that if $\chi_{\pi'}$ is created with a different permutation $\pi'\neq\pi$, 
then the answer will be a random variable with distribution close to uniform 
(see Theorem~\ref{theo:correctness} for details).

Algorithm~\ref{alg:distinction} provides a computational approach to the problem of distinguishing the two states and is suitable for defining the public key scheme presented in \cite{kaw:kos:nis:yam:05}.
In our work, due to the nature of the problem at hand, we will instead consider measurements.
We will show next that the result of Algorithm~\ref{alg:distinction} is equivalent to measuring the orthogonal observable 
$$
M_\pi = 0 \cdot P_\pi^+ + 1 \cdot P_\pi^-
$$ 
where 
$$ P_\pi^\pm 	=\sum_{\sigma \in \E_n} \ket{\psi^\pm_{\pi} (\sigma)}\bra{\psi^\pm_{\pi} (\sigma)} 
			= \frac{1}{2} \sum_{\sigma \in \E_n}(\ket\sigma \pm \ket{\sigma \circ\pi})(\bra\sigma \pm \bra{\sigma \circ\pi}).
$$

\begin{proposition}\label{prop:measuredistinction}
	Applying Algorithm \ref{alg:distinction} with $\pi$ and measuring $M_\pi = 0 \cdot P_\pi^+ + 1 \cdot P_\pi^-$ are equivalent processes:
	the probability distribution of the outcomes of the Algorithm~\ref{alg:distinction} 
	is the same as the probability distribution of the outcomes of the measurement 
	$M_\pi$, and
	the resulting states are the same.
\end{proposition}

As in Algorithm~\ref{alg:distinction}, we consider a fixed $n$ and the case where $k(n)=1$. 
The reader can easily generalize the argument for $k(n)$ systems using $M_\pi^{\otimes k(n)}$ instead of $M_\pi$.

\begin{proof}
In order to prove this proposition, first notice that by the linearity of the measurement, it is enough to show the result only for the case of a general pure $\ket \psi$. 
Since for every $\pi$ the set 
		$\{\ket{\psi^\pm_{\pi} (\sigma)}:\sigma\in\E_n\}$
is a basis, we can write
$$
\ket\psi = \sum_{\sigma\in\E_n} \left(c_\sigma^+\ket{\psi^+_{\pi} (\sigma)} + c_\sigma^-\ket{\psi^-_{\pi} (\sigma)} \right).
$$

We first compute the following:
$$
\begin{array}{rcl}
	\ds P_\pi^\pm\ket\psi
		&=&  \ds	\sum_{\sigma\in\E_n} \ket{\psi^\pm_{\pi} (\sigma)} \bra{\psi^\pm_{\pi} (\sigma)} 
				\left(\sum_{\sigma'\in\E_n} \left(c_{\sigma'}^+\ket{\psi^+_{\pi} (\sigma')} +
				 c_{\sigma'}^-\ket{\psi^-_{\pi} (\sigma')} \right)\right)\\[10pt]
%
%
		&=& \ds \sum_{\sigma\in\E_n} \sum_{\sigma'\in\E_n}\ket{\psi^\pm_{\pi} (\sigma)}  
				\left(c_{\sigma'}^+\bracket{\psi^\pm_{\pi} (\sigma)}{\psi^+_{\pi} (\sigma')} + 
				c_{\sigma'}^-\bracket{\psi^\pm_{\pi} (\sigma)}{\psi^-_{\pi} (\sigma')}\right)\\[10pt]
		&=& \ds\sum_{\sigma\in\E_n} c_{\sigma}^\pm\ket{\psi^\pm_{\pi} (\sigma)}.
\end{array}$$
The result of the first $4$ Steps of the Algorithm \ref{alg:distinction} is ($\ket{\psi^\pm_{\pi} (\sigma)}= \pm\ket{\psi^\pm_{\pi} (\sigma\pi)}$):
$$\begin{array}{rcl}
  \multicolumn{3}{l}{(H\otimes \I) C_\pi (H  \otimes \I )\ket 0\ket\psi} \nonumber\\[10pt]
	&=& \ds \frac{1}{\sqrt 2}(H\otimes \I) C_\pi ( \ket 0 + \ket 1)\ket{\psi}\nonumber\\[10pt]
	&=& \ds \frac{1}{\sqrt 2}(H\otimes \I)C_\pi \left(\ket 0\ket{\psi} + \ket 1\sum_{\sigma\in\E_n} \left(c_\sigma^+\ket{\psi^+_{\pi} (\sigma)} + c_\sigma^-\ket{\psi^-_{\pi} (\sigma)} \right)\right)\nonumber\\[10pt]
	&=&  \ds \frac{1}{\sqrt 2}(H\otimes \I) \left(\ket 0\ket{\psi} + \ket 1\sum_{\sigma\in\E_n} \left(c_\sigma^+\ket{\psi^+_{\pi} (\sigma)} - c_\sigma^-\ket{\psi^-_{\pi} (\sigma)} \right)\right)\nonumber\\[10pt]
	&=&\ds \frac{\ket 0 + \ket 1}{2}\sum_{\sigma\in\E_n} \left(c_\sigma^+\ket{\psi^+_{\pi} (\sigma)} + 
			c_\sigma^-\ket{\psi^-_{\pi} (\sigma)} \right)\\ [5mm]& & \ds\hspace{10mm} + \frac{\ket 0 - \ket 1}{2}\sum_{\sigma\in\E_n} 
			\left(c_\sigma^+\ket{\psi^+_{\pi} (\sigma)} - c_\sigma^-\ket{\psi^-_{\pi} (\sigma)} \right)\nonumber\\[10pt]
	 &=&\ds \ket 0 \sum_{\sigma\in\E_n} c_\sigma^+\ket{\psi^+_{\pi} (\sigma)} +
			 \ket 1 \sum_{\sigma\in\E_n} c_\sigma^-\ket{\psi^-_{\pi} (\sigma)}\label{eq:eq1}\\[10pt]
	 &=& \ds \ket 0 P_\pi^+\ket\psi + \ket 1 P_\pi^-\ket\psi.	
\end{array}$$
Measuring the first register in the computational basis (Step \ref{step:5:alg} of Algorithm \ref{alg:distinction}) collapses the second register (up to a normalization factor) in the state 
		$P_\pi^\pm\ket\psi$ 
with probability $||P_\pi^\pm\ket\psi||^2$, 
just as if $M_\pi$ was measured.
%
\qed
\end{proof}

An immediate corollary is that $M_\pi$ distinguishes  
		$\rho^+_\pi$ and $\rho^-_\pi$ 
with probability one.


At the end of an oblivious transfer protocol, Bob must know if he received the message or not. 
In our protocol, this step is guaranteed by a {\it universal hash function}. 
These functions map larger strings to strings of smaller size, hence collisions are unavoidable, \textit{i.e.}, different messages can be mapped to the same hash. 
Despite this fact, one can design these functions in such a way that:
\begin{itemize}
	\item they are computationally efficient, \textit{i.e.}, one can compute its value from a message in polynomial time;
	\item hashes are almost equally distributed.
 \end{itemize}

\begin{definition}
	Let $A$ and $B$ be two sets of size  $a$ and $b$, respectively, such that 
		$a>b$, and let $\mathbbm H$ be a collection of hash functions $h:A\to B$.
	$\mathcal H$ is said to be a \emph{universal family of hash functions} if 
			$$\Pr_{h\in \mathbbm H}[h(x) = h(y)] \leq\frac 1 b.$$
\end{definition} 

A straightforward consequence of this definition is the following theorem.
\begin{theorem}\label{theo:universalhash}
	Let $A$ and $B$ be two sets of size $a$ and $b$, respectively, 
	such that $a>b$ and let $\mathbbm H$ be a collection of hash functions $h:A\to B$. 
	If $\mathbbm H$ is a universal family of hash functions then for any set $A' \subset A$ of size $N$ and for any $x\in A$, 
	the expected number of collisions between $x$ and other elements in $A'$ is at most $N/b$.
\end{theorem}

Notice that, in particular, if we request $A$ to contain all strings of length $\ell$ and $B$ to have length $\ell/2$, then the number of expected collisions is $2^{\ell/2}$, hence the probability of finding a collision is negligible in $\ell$.
There are several standard ways to construct universal families of hash functions (see  \cite{car:weg:79} for examples and details). 

\section{The oblivious transfer protocol}\label{sec:protocol}

In this section, we present a quantum protocol that achieves oblivious transfer from Alice to Bob in polynomial time. 
As already mentioned in the introduction, oblivious transfer protocol is a protocol in which Alice sends (transfers) a message
		 $m = m_1m_2\ldots m_\ell$ of length $\ell$ 
to Bob (during the transfer phase), which is recovered by him with only 
	$50\%$ of probability during the opening phase (the {\em transfer is probabilistic}). 
The protocol must satisfy two additional properties: 
	be {\em concealing}, \textit{i.e.}, Bob must not learn $m$ before the opening phase;
	and be {\em oblivious}, \textit{i.e.}, after the opening phase Alice must {\em not know with certainty} if Bob received $m$ or not.
Notice that Bob, unlike Alice, {\em does know} at the end of the protocol if he received the intended message. In the rest of this section, we present our results in terms of orthogonal measurements $M_\pi$, which according to Proposition~\ref{prop:measuredistinction} can be realized in polynomial time by applying Algorithm~\ref{alg:distinction}.

\begin{protocol}[Oblivious transfer]\label{prot:ot}
\  
\begin{quote}
	\begin{description}
		\item[Message to transfer] $m=m_1\dots m_\ell$.
		\item[Security parameter] $n$.
		\item[Universal hash function] $h:\Sigma^\ell\to \Sigma^{\ell/2}$.
		\item[Secret key] $\pi\in \K_ n$.
		\item[]\vspace*{-10pt}
		\item[Transfering phase:] \
			\begin{senum}
				\item Alice generates uniformly at random the secret key $\pi\in \K_n$.
				\item Using Algorithm~\ref{alg:keysgeneration} with $\pi$, and the operation $\csignal$, she constructs the state 
					$\rho_\pi^m = \rho_\pi^{m_1}\otimes \rho_\pi^{m_2}\otimes\ldots\otimes \rho_\pi^{m_\ell}$ 
				where 
					$ \rho_\pi^{m_i} = \rho_\pi^{+}$ if $m_i=1$, 
				and 
					$ \rho_\pi^{m_i} = \rho_\pi^{-}$ otherwise.
				\item Alice sends $\rho_\pi^m$ and $y=h(m)$ to Bob.
			\end{senum}
		\item[]\vspace*{-10pt}
		\item[Opening phase:] \
			\begin{senum}
				\setcounter{cntsenum}{3}
				\item Bob generates uniformly at random $\tau\in \S_n$ and sends it to Alice.
	
				\item Alice computes uniformly at random either 
						$\delta= \pi\circ \tau$ or $\delta=\tau\circ \pi$, 
					and sends it back to Bob.\label{prot:step2}
	
				\item Bob computes uniformly at random either 
						$\gamma= \delta\circ\tau^{-1}$ or $\gamma= \tau^{-1}\circ\delta$.\label{prot:step3}
	
				\item  Bob measures the observable $M_\gamma^{\otimes \ell}$ on the system given by Alice, obtaining the result $\tilde{m}$.\label{prot:step7}
	
				\item Bob checks if $h(\tilde{m}) = y$. 
						If so, he concludes that the message sent by Alice is $\tilde{m}$, {\it i.e.}, $\tilde m = m$.\label{prot:step5}
	
				\item If Bob got the correct message, 
						he chooses another $\pi'\in\K_n$ and measures the observable 
								$M_{\pi'}^{\otimes \ell}$ on the system, obtaining result 
								$r=r_1\dots r_\ell$. 
					If approximately half of the results $r_i$ are different from the corresponding $m_i$, 
					then Bob accepts the message; otherwise he aborts the protocol declaring that
					Alice tried to cheat. \label{prot:step9}
		\end{senum}
	\end{description}
\end{quote}
\end{protocol}

To show that our proposal is an oblivious transfer protocol we must prove that:
\begin{enumerate}
		\item If Alice is honest, then the protocol is computationally \emph{concealing}, 
				\textit{i.e.}, Bob cannot learn the message $m$ before the opening phase. 
			Notice that this follows directly from the hardness assumption of computational indistinguishability of states $\rho^+$ and $\rho^-$ (Theorem \ref{theo:hardness}) and from the fact that the hash function used for comparison has exponentially many collisions, \textit{i.e.}, only with negligible probability Bob can correctly invert $h$ in probabilistic polynomial time and obtain $m$.

		\item If both Alice and Bob are honest, 
				\textit{i.e.}, they play their roles accordingly to Protocol~\ref{prot:ot} to transfer message $m$, 
			then roughly in $50\%$ of the cases Bob will obtain~$m$ ({\em probabilistic transfer}). 
			Notice that the acknowledgment that $\tilde m$ is correct is given by the comparison of 
				$y$ with $h(\tilde m)$;

		\item If Bob is honest, then the protocol is \emph{oblivious}, 
				\textit{i.e.}, Alice cannot learn with certainty whether  Bob got the message $m$ or not. 
			This is a consequence of the impossibility of faster than light information transmission and it is ensured by the last step of the protocol. The \emph{rationale} for Step~\ref{prot:step9} is to prevent Alice from cheating by sending a state that would allow her to know with certainty that Bob would receive the message. We postpone the discussion of this point for the end of the paper. 
\end{enumerate}

Notice that if the order of compositions agree in Steps~\ref{prot:step2} and~\ref{prot:step3} of the Protocol, 
\textit{i.e.}, Alice and Bob applied respectively $\tau$ and $\tau^{-1}$ on the same side, then $\gamma=\pi$; 
otherwise, if Alice and Bob applied $\tau$ and $\tau^{-1}$ on different sides, then both 
$\gamma=\tau^{-1}\circ\pi\circ\tau$ and $\gamma=\tau\circ\pi\circ\tau^{-1}$ belong to $\K_n$ but are different from $\pi$.

Also, due to the properties of hash functions, two messages having the same hash are hard to find, in the sense that the probability of such event is negligible. 
Therefore, if in Step~\ref{prot:step5} of the Protocol $h(\tilde{m}) = y = h(m)$, then $\tilde m$ is, up to negligible probability, the actual message $m$. 

Therefore, if both Alice and Bob are honest, then Bob measures with equal probability either 
		$M_{\pi}^{\otimes\ell}$ or $M_{\pi^\prime}^{\otimes\ell}$, 
where $\pi^\prime = \tau\circ\pi\circ\tau^{-1} \in \K_ n$ for some $\tau\in\S_n$. 
If he measures $M_{\pi}^{\otimes\ell}$, then the measurement result $\tilde{m}$ will indeed be equal to $m$, which he can confirm by comparing $h(\tilde{m})$ and $y$. 
Otherwise, and since $\pi'$ does not match the used $\pi$, each single-system measurement will yield a random 
	$\tilde{m}_i$, 
and by confirming that $h(\tilde{m}) \neq y$, 
Bob will know that he did not receive the intended message. 

This way, if both parties are honest, Bob will recover the message with probability~$1/2$. We formalize these in the following theorem:

\begin{theorem}[Correctness of the Protocol \ref{prot:ot}]
\label{theo:correctness}
		Assume that $\QSCDff$ is polynomially hard even for quantum computers. If Alice and Bob correctly run Protocol~\ref{prot:ot} to transfer message $m=m_1\dots m_\ell$ 
		from Alice to Bob, then:

			\begin{enumerate}
				\item {\em{\bf(Concealing) }} Bob cannot infer $m$ before the opening phase except with negligible probability. 
				
				\item {\em{\bf(Probabilistic transfer) }} Bob will receive $m$ with probability 
						$1/2+\varepsilon(\ell)$, where $\varepsilon(\ell)$ is a negligible function.%
						\footnote{A function $\varepsilon(\ell)$ is said to be negligible if  $\varepsilon(\ell)<1/p(\ell)$ for any polynomial $p(\ell)$ and sufficiently large~$\ell$.}

				\item {\em{\bf(Oblivious) }}Alice remains oblivious to the fact that Bob received the message.
			\end{enumerate}
\end{theorem}

\begin{proof}
We prove each item stated in the theorem separately.
\begin{enumerate}
\item
The concealing property follows directly form Theorem \ref{theo:hardness} proved in \cite{kaw:kos:nis:yam:05}, which states that distinguishing $\rho_\pi^+$ from $\rho_\pi^-$ without knowing $\pi$ is polynomially hard even for a quantum computer. 
%
%
Notice that since Alice does not send $\pi$, the probability that Bob guesses the correct $\pi$ is negligible: the size of $\K_n$ is  
		$$|\K_n| = \ds{\frac{n!}{(n/2)!}}\; ,$$ 
which is already for $n=10$ huge enough for practical purposes.
On the other hand, since it is assumed that $h$ is a universal hash function, trying to recover~$m$ by computing the inverse of $h$ is only possible with negligible probability. 
In fact, since the hash function $h$ maps strings of length $\ell$ to strings of length $\ell/2$, by Theorem~\ref{theo:universalhash} one can recover $m$ with probability $2^{-\ell/2}$ which is negligible in $\ell$.

\item
It is easy to see that after Step~\ref{prot:step3} of Protocol~\ref{prot:ot},
$\gamma=\pi$ with probability $1/2$ if both Alice and Bob make their choices in Steps~\ref{prot:step2} and~\ref{prot:step3} at random, and Bob chooses $\tau$ uniformly at random. 

For $\gamma=\pi$, the result that the measurement 
	$M_\pi^{\otimes \ell}(\rho^m_\pi)$ equals $m$ 
follows from Proposition~\ref{prop:measuredistinction}.

It remains to show that when 
		$\gamma=\pi'\neq\pi$,
by performing the measurement $M_\gamma^{\otimes \ell}$,
the probability of recovering 
		$m$ 
from the quantum state 
	$\rho^m_\pi$ 
sent by Alice is negligible. 
First we prove that for any $\sigma$ the result of the measurement of $M_{\pi'}^\ell$ is random. 
Observe that
$$
\hspace{-20mm}\begin{array}{rcl}
	\ds P_{\pi'}^\pm(\ket{\psi^+_\pi(\sigma)}) 
	&=& 
		\ds \sum_{\sigma' \in \E_n}
			\ket{\psi^{\pm}_{\pi'}(\sigma')} 
			\bracket{\psi^{\pm}_{\pi'}(\sigma')}
			{\psi^{+}_{\pi}(\sigma)}\\[10pt]
	&=&
		\ds\frac{1}{2\sqrt 2} \sum_{\sigma' \in \E_n}
			(\ket{\sigma'} \pm \ket{\sigma' \circ\pi'})
			(\bra{\sigma'} \pm \bra{\sigma' \circ\pi'})
			(\ket\sigma + \ket{\sigma\circ\pi})\\[10pt]
	&=& 
		\ds\frac{1}{2\sqrt 2} \sum_{\sigma' \in \E_n}
			(\ket{\sigma'} \pm \ket{\sigma' \circ\pi'})
			(\bracket{\sigma'}{\sigma} +
				\bracket{\sigma'}{\sigma\circ\pi} \pm
				\bracket{\sigma'\circ\pi'}{\sigma} \pm
				\bracket{\sigma'\circ\pi'}{\sigma\circ\pi})\\[10pt]
	&=& 
		\ds 
			\left\{
				\begin{array}{rcl}
					\ds\frac{1}{2\sqrt 2} \sum_{\sigma' \in \E_n}
						(\ket{\sigma'} \pm \ket{\sigma' \circ\pi'})
						(\bracket{\sigma'}{\sigma} +
						\bracket{\sigma'\circ\pi'}{\sigma\circ\pi}),  
												& & \mbox{ if } \sigma \in \E_n\\
					\ds\frac{1}{2\sqrt 2} \sum_{\sigma' \in \E_n}
						(\ket{\sigma'} \pm \ket{\sigma'\circ \pi'})
						(\bracket{\sigma'}{\sigma\circ\pi} \pm
						\bracket{\sigma'\circ\pi'}{\sigma}), 
												& & \mbox{ if } \sigma \in \O_n
				\end{array}
			\right.\\[10pt]
	&=&
		\ds \frac{1}{2\sqrt 2}
			\left(
				\ket{\sigma} \pm \ket{\sigma \circ\pi'} + 
				\ket{\sigma\circ\pi} \pm \ket{\sigma\circ\pi\circ \pi'}
			\right)\\[10pt]
	&=&
		\ds \frac 1 2 (\ket{\psi^{\pm}_{\pi'}(\sigma)} \pm \ket{\psi^{\pm}_{\pi'}(\sigma\circ\pi\circ\pi')}\\[10pt]
	&=& 
		\ds \frac{1}{\sqrt 2}\left(\frac{1}{\sqrt 2} 
				(\ket{\psi^{\pm}_{\pi'}(\sigma)} \pm\ket{\psi^{\pm}_{\pi'}(\sigma\circ\pi\circ\pi')})  \right).
\end{array}
$$

Since
	$\ket{\psi^{\pm}_{\pi'}(\sigma)} $ and $\ket{\psi^{\pm}_{\pi'}(\sigma\circ\pi\circ\pi')}$ 
are  orthogonal, the vector 
	$ \frac{1}{\sqrt 2} (\ket{\psi^{\pm}_{\pi'}(\sigma)} \pm\ket{\psi^{\pm}_{\pi'}(\sigma\circ\pi\circ\pi')})$  
is unitary and $||P_{\pi'}^\pm(\psi^+_\pi)||^2 = 1/2$.
Hence, the probability of recovering $\pm$ from $\rho_\pi^+$ is
$$\begin{array}{rcl}
Prob	 (+; M_{\pi'}^{\otimes \ell} , \rho_\pi^+)	
		&=& \ds Tr[P_{\pi'}^+\rho_\pi^+P_{\pi'}^+]\\[10pt]
		&=& \ds \frac{2}{n!}\sum_{\sigma\in \E_n} Tr[P_{\pi'}^+\ket{\psi^{+}_{\pi'}(\sigma)}\bra{\psi^{+}_{\pi'}(\sigma)} P_{\pi'}^+]\\[10pt]
		&=& \ds \frac{2}{n!}\sum_{\sigma\in \E_n}||P_{\pi'}^+(\psi^+_\pi)||^2\\[10pt]
		&=& \ds \frac 1 2 
\end{array}
$$
and similarly 
$Prob (-; M_{\pi'}^{\otimes \ell} , \rho_\pi^+) = \frac 1 2$. 
{\it Mutatis mutandis } we also have $Prob (\pm; M_{\pi'}^{\otimes \ell} , \rho_\pi^-) = \frac 1 2$.
Hence, by measuring $M_{\pi'}^{\otimes \ell}$ on $\rho^m_\pi$ the probability of recovering $m$ is negligible in $\ell$. 

To conclude the proof of the second item we need to show that Bob aborts the protocol in Step~\ref{prot:step9} only with negligible probability. 

Notice that to reach Step~\ref{prot:step9}, where Bob aborts the protocol, he must had run successfully the verification in Step~\ref{prot:step7}.
This implies that Bob performed the measurement with the correct trapdoor $\pi$, hence the state stayed invariant, i.e., it is still $\rho_\pi^m$. 
If Bob chooses in Step~\ref{prot:step9} random $\pi'\neq \pi$, then the probability of recovering each bit of the message, as just seen above, is equal to $1/2$ and so, by a simple binomial argument, the probability of having a significant difference from half of the states is negligible on $\ell$.

\item 
Notice that Alice must send a permutation $\delta$, such that both 
	$\delta\circ\tau^{-1}$ and $\tau^{-1}\circ\delta$ 
are from $\K_ n$. 
Bob's choice to compose $\delta$ with $\tau^{-1}$ on the left, or on the right, is random and unknown to Alice: 
after sending $\delta$ there is no more communication between Alice and Bob and there is no information transmission from Bob to Alice. 
Therefore, the choice of Bob's measurement observable 
	$M_\gamma^{\otimes\ell}$ 
is also unknown to Alice as well: Alice cannot know if Bob has obtained the message $m$, or not. \qed
\end{enumerate}
\end{proof}

Finally we prove the security of the protocol against cheating strategies of Alice and Bob.
\begin{theorem}[Security] Assume that $\QSCDff$ is polynomially hard even for a quantum computer. The Protocol \ref{prot:ot} is secure against cheating, {i.e.} 
\begin{enumerate}
	\item {\em{\bf(Concealing) }}  If Alice is honest, then Bob cannot learn $m$ before the opening phase even if he cheats.
	\item {\em{\bf(Oblivious) }}  If Bob is honest, then Alice cannot learn with certainty if he received the message even if she tries to cheat.
\end{enumerate}
\end{theorem}
\begin{proof}
\begin{enumerate}
\item 
Notice that upon receiving a system from Alice the only thing Bob can do in order to learn its state (and hence the message $m$) is to simply ``look'' at it, {\it i.e.}, perform a measurement in order to distinguish between $\rho_\pi^+$ and  $\rho_\pi^-$ (without knowing $\pi$). This is exactly what an honest Bob could do, which was proved in Theorem \ref{theo:correctness} to be unfeasible.

\item
To finish the argument of the security of Protocol \ref{prot:ot} we show that  the protocol is unconditionally \emph{oblivious} against a cheating Alice, \textit{i.e.}, there is no strategy for Alice which would enable her to know with certainty if Bob received the message~$m$ or not. 
This is ensured by the last step of the protocol.
Notice that since Alice cannot know beforehand which permutation $\tau$ will be chosen by Bob, nor she knows which~$\gamma$ Bob choses in Step~\ref{prot:step3} of the protocol, she does not know which measurement $M_\gamma^{\otimes \ell}$, with $\gamma\in\K_n$, is performed by Bob in Step~\ref{prot:step7}. 
Therefore, in order to know if Bob received the message $m$, she must prepare a state $\rho^m$  that leads to the same answer $m$ regardless of the measurement selected by Bob.
Obviously, in order to satisfy the requirement that Bob learns $m$ in $50\%$ of the cases, she sends uniformly at random either $\rho^m$, in which case she knows with certainty that Bob got the message, or a completely mixed state $(\mathbbm 1/n!)^{\otimes \ell}$, in which case she knows with certainty that Bob does not get the message except with negligible probability.
So, if Alice wants to be non-oblivious, she needs to prepare states that 
give with certainty the same result for every measurement $M_\pi$, and are thus invariant,
which is the reason for introducing Step~\ref{prot:step9}.
If Bob got the correct message, he can recheck that Alice did not use this cheating strategy by performing another measurement of the same kind but with a different $\pi'$. 
Since the state sent by Alice has to be invariant for any measurement, measuring the state with this $\pi'$ will lead to the same result as the original choice and Bob will thus abort the protocol. \qed
\end{enumerate}
\end{proof}

\section{Conclusions}\label{sec:conclusion}
Oblivious transfer is an important primitive for designing cryptographic protocols and secure multiparty computation schemes.
In this paper we proposed a polynomial time quantum protocol for oblivious transfer of information from Alice to Bob based on the $\QSCDff$ state distinguishability problem. 
We showed that, assuming $\QSCDff$ to be polynomially hard even for a quantum computer, our protocol is computationally concealing, oblivious, achieves the goal of transferring information with probability close to $1/2$, and is secure against cheating strategies.
The oblivious and the probabilistic transferring properties rely on the laws of quantum mechanics, while the acknowledgment of the message transfer is ensured by the use of hash functions. 
Note that, in general, one may not need to use hash functions. 
For example, if the message sent by Alice is of a particular type, say an $NP$-problem ({\it e.g.} SAT or some hard optimization problem), then there is no need to encode the message in the hash value: 
Bob can verify if he received the message by simply checking if it is the solution of the problem. 
Also if the message sent by Alice is a binary code of some text in a human language, then the verification of the meaningful information serves as acknowledgment.

\section*{Acknowledgments}
This work was partially supported by FCT projects 
		QSec PTDC/EIA/67661/2006,
		QuantPrivTel PTDC/EEATEL/ 103402/2008,
		ComFormCrypt PTDC/EIA-CCO/ 113033/2009,
		PEst-OE/EEI/LA0008/2013, 
and SQIG's initiative PQDR (Probabilistic, Quantum and Differential Reasoning), QuantTel and `P-Quantum', as well as Network of Excellence, Euro-NF.

Andr\'e Souto also acknowledges the FCT postdoc grant SFRH/BPD/76231/2011.

\end{document}